\documentclass[journal]{IEEEtran}
\usepackage{subfigure,graphicx}
\ifCLASSINFOpdf
\else
\fi
\usepackage{amsmath,amsthm,amssymb}
\newtheorem{theorem}{\bf Theorem}

\newtheorem{corollary}{\bf Corollary}

\newtheorem{proposition}{\bf Proposition}
\newtheorem{lemma}{\bf Lemma}

\usepackage{booktabs}
\usepackage{multirow,multicol}
\usepackage{xcolor}
\begin{document}
%
\title{An Insurance Contract Design to Boost Storage Participation in the Electricity Market
}
%
%
%


\author{Nayara~Aguiar
        and~Vijay~Gupta
\thanks{The authors are with the Department of Electrical Engineering, University of Notre Dame, Notre Dame, IN 46556 USA
        (email: ngomesde@nd.edu, vgupta2@nd.edu).}
\thanks{This work was supported in part by NSF under grant CNS-1739295 and Sandia National Lab PO 2079716.}
}

\maketitle

\begin{abstract}
Energy storage technologies are key to improving grid flexibility in the presence of increasing amounts of intermittent renewable generation. We propose an insurance contract that suitably compensates energy storage systems for providing flexibility. Such a contract provides a wider range of market opportunities for these systems while also incentivizing higher renewable penetration in the grid. We consider a day-ahead market in which generators, including renewables and storage owners, bid to be scheduled for the next operating day. Due to production uncertainty, renewable generators may be unable to meet their day-ahead production schedule, and thus be subject to a penalty. As a hedge against these penalties, we propose an insurance contract between a renewable producer and a storage owner, in which the storage reserves some energy to be used in case of renewable shortfalls. We show that such a contract incentivizes the renewable player to bid higher, thus increasing renewable participation in the electricity mix. It also provides an extra source of revenue for storage owners that may not be profitable with a purely arbitrage-based strategy in the day-ahead market. 
We validate our analysis through two case studies.
\end{abstract}


%
\IEEEpeerreviewmaketitle

\vspace{-0.5em}
\section{Introduction}
%
%
%
%
The fast ramp capability of storage technologies makes them ideal candidates to improve system reliability as renewable penetration in the power grid deepens \cite{divya}. In recent years, the most cost-effective entry point for storage operators in the electricity market has been in the ancillary market \cite{xu2016}. However, this market is relatively small as compared to those that cover other grid services and may start to saturate as more energy storage systems enter the market. For instance, there is operational evidence that the PJM RegD market has become saturated \cite{pjm_stateofmarket}. Therefore, new market opportunities are needed to incentivize storage capacity to increase in the grid. 

One potential option is the peaking capacity market. In California alone, a total of 13GW of peaking capacity provided by conventional generators is expected to retire within the next two decades, while the frequency regulation market in the entire United States is estimated at 5GW \cite{denholm}. However, at current storage prices, arbitrage to provide energy in the energy markets based on mechanisms such as time of usage prices are rarely cost-effective~\cite{kalathil2017sharing}. We argue that market mechanisms in which storage operators are compensated for their provision of flexibility rather than merely arbitrage to provide energy {\em on demand} are needed. Such a mechanism will have two benefits. One, by providing extra revenue to storage owners, it will boost investments into storage technology even further. Two, as storage becomes more economical, it will lead to a higher penetration of intermittent renewable energies into the grid, while ensuring reliability of electricity supply.  

In this paper, we propose the design of a mechanism on the lines of an insurance contract between a storage owner and a renewable producer. Through this contract, the storage commits to reserve some energy to be used in case of renewable shortage, while the renewable purchases the right to call upon this reserve if needed. The contract is only feasible when there exists a reserve price for which both participants voluntarily agree to sign it. 
We show that the insurance contract incentivizes the renewable producer to bid higher in the day-ahead market. We also derive a condition for which a storage unit is profitable as an insurance provider, even if it is not profitable in the day-ahead market. Combined, these results prove that the insurance contract proposed is an opportunity to boost storage participation in the market, while increasing the share of renewable energy in the electricity mix. 

This insurance contract can be viewed as lying between the two extremes of the ISO incurring the increasing cost of system reserves in a grid-takes-all-renewable scenario and the renewable incurring the entire cost of its intermittency. In practice, there is already a move towards asking renewable producers to shoulder (some) of the responsibility of hedging against their own variability. For example, the Bonneville Power Administration (BPA) has a self-supply program in which variable energy resources are allowed to supply their own balancing services in lieu of incurring the cost of the services provided by BPA \cite{bpa}. The mechanism we propose is a step towards putting such ad hoc arrangements on a firmer analytical footing. The novelty of our approach consists of: (i) requiring the renewable producer to incur his own uncertainty cost, leading to a scenario where this cost is borne by the agent who causes the electricity supply to become less reliable. This is in contrast to most market operations, which socialize reserve costs among, for example, load serving entities; (ii) modeling the renewable generator and the storage as independent players that aim to find a contract that can benefit them both. Previous works have largely considered these energy sources to be operated jointly or that they are standalone players that do not interact; and (iii) providing analytical and numerical results which bring insights about how the proposed contract serves as an extra economic incentive both for storage technologies to enter the market and for renewables to increase their participation in the electricity mix.

We focus on a stylized two-settlement market composed of a day-ahead market and an ex-post imbalance resolution mechanism. In the day-ahead market, suppliers offer their production for delivery in the next operating day, and their production schedule is settled by an Independent System Operator (ISO). Generators that deviate their actual production from their schedule are penalized. These policies have traditionally been imposed for non-renewable, firm generators; however, the adoption of the same treatment for intermittent renewable generators has also started recently, as grids move away from a take-all-renewables approach \cite{porter}. However, levying of such penalties may lead to renewables bidding conservatively in the market \cite{bitar}. The proposed insurance contract can counteract this undesirable decrease in renewable participation, while properly compensating sources of flexibility.

A related stream of work analyzes how independent storage operators can take advantage of arbitrage opportunities in the market \cite{krishnamurthy,wangy,xu}. There are numerous studies regarding the simultaneous participation of storage in multiple services, such as energy, reserve, and ancillary services \cite{baggu,pandzic,zou}. Although exploring multiple market services allows for a full overview about the revenue potential for storage, it comes at the expense of tractability. This may lead to complex models being solved numerically and a lack of insights about the results achieved. Thus, we consider only the day-ahead energy market. A similar approach has been used, e.g., in \cite{mcconnell}, in which the adoption of cap contracts by storage systems as a hedge against price volatility is analyzed.

Several works have analyzed the optimal placement and sizing of energy storage in the grid \cite{baker,ghofrani,thrampoulidis}. While we do not focus on the placement problem, we evaluate how this aspect plays a role in the adoption of our insurance contract in a case study in which the location of the storage in the grid is varied. The cooperation between storage and renewable sources has also been explored at length in the literature, as in \cite{kim,su,teleke}. These schemes are often approached as a joint operation or a centralized problem, and our formulation departs from this paradigm by considering the energy sources to be independent agents maximizing their own profit. Lastly, risk factors due to renewable variability are included in a variety of previous works. For instance, \cite{kaabeche} designs a hybrid system with considerations to the probability of having insufficient supply, \cite{asensio} analyses the pairing of a risk-averse wind producer and demand response in the day-ahead market, and \cite{li} considers uncertainty when modeling spinning reserve requirements for a microgrid. In our formulation, uncertainty stemming from renewable production influences the price of the insurance contract, as the storage bears more risk when the likelihood of a renewable shortage is high.

\paragraph*{Paper Organization}
The mathematical models are introduced in Section~\ref{sec:models}. Section~\ref{sec:problem} presents the utility functions of the participants, and the insurance contract design problem. In Section~\ref{sec:results}, we derive the strategies of the participants in the day-ahead market, and their conditions to sign the proposed contract, which are used to analyze the contract's feasibility and the profitability of the storage. Section~\ref{sec:casestudy} presents two case studies, followed by conclusions and directions for future work in Section~\ref{sec:conclusions}. The main proofs are given in the Appendix.

\vspace{-0.5em}
\section{Supply and Electricity Market Models}\label{sec:models}

\textit{Energy Storage Model:}\label{sec:storagemodel}
Consider a time horizon divided into $N$ discrete time slots indexed by $k \in \mathcal{K} :=\{0,1,...,N-1\}$. The state of charge of the storage $x_{k}$ is the amount of energy stored at the beginning of time $k$. 
We consider a perfectly efficient storage whose state of charge follows the dynamics
\begin{equation}\label{eq:dynamics0}
     x_{k+1}= x_k-u_k ~ \forall k \in \mathcal{K},
\end{equation}
where $u_{k}$ is the energy extracted from or injected into the storage, which is positive (resp. negative) if the storage is discharging (resp. charging). The input $u_{k}$ can be denoted as the difference between the discharged $u^+_{k} \geq 0$ and the charged $u^-_{k} \geq 0$ quantities at time $k$. Then, \eqref{eq:dynamics0} can be rewritten as 
\begin{equation}\label{eq:dynamics}
     x_{k+1}= x_{k}-u^+_{k}+u^-_{k} ~ \forall k \in \mathcal{K}.
\end{equation}

The storage state and inputs are constrained by
\begin{equation}\label{eq:S_constr1}
    0 \leq x_k \leq \overline{E}~ \forall k \in \mathcal{K}
\end{equation}
\begin{equation}\label{eq:S_constr2}
    u^+_{k}u^-_{k} = 0 ~ \forall k \in \mathcal{K}
\end{equation}
\begin{equation}\label{eq:S_constr3}
    u^+_{k}, u^-_{k} \geq 0 ~ \forall k \in \mathcal{K}.
\end{equation}
The constraint \eqref{eq:S_constr1} sets the bounds on the amount of energy that can be stored in the storage unit, \eqref{eq:S_constr2} is a complementarity constraint which prevents simultaneous charging and discharging, and \eqref{eq:S_constr3} is a positivity constraint. Using the storage dynamics \eqref{eq:dynamics}, we can write \eqref{eq:S_constr1} recursively to find the compact form
\begin{equation}\label{eq:S_constr1c}
    \mathbf{0 \leq \mathbf{x_0} +A^+u^+ + A^-u^- \leq \overline{E}},
\end{equation}
where $\mathbf{A^+},~\mathbf{A^-} \in \mathbb{R}^{N\times N}$ are triangular matrices with $A^+_{ij}=-1$ and $A^-_{ij}=1$ for all $i\geq j$. The column vectors $\mathbf{u^+},~\mathbf{u^-} \in \mathbb{R}^{N}$ are defined as $\mathbf{u^+} = [u^+_0,...,u^+_{N-1}]^T$ and $\mathbf{u^-} = [u^-_0,...,u^-_{N-1}]^T$. Further, $\mathbf{0}$ is the null column vector of size $N$, and $\mathbf{\overline{E}} = \overline{E} \mathbf{1^T}$, where $\mathbf{1^T} \in \mathbb{R}^N$ is the all-ones column vector. The initial condition $\mathbf{x_0}=x_0\mathbf{1}^T$ accounts for the initial state of charge of the storage, and can be removed if the device is initially discharged. 
Note that the inequalities in \eqref{eq:S_constr1c} are element-wise.
Let $\mathcal{U}$ denote the set of all pairs $(\mathbf{u^+},\mathbf{u^-})$ that satisfy the storage constraints \eqref{eq:dynamics} -- \eqref{eq:S_constr3}. A storage policy $(\mathbf{u^+},\mathbf{u^-})$ is said to be feasible if $(\mathbf{u^+},\mathbf{u^-}) \in \mathcal{U}$.

We initially ignore the power constraint of the storage,
which restricts its ramp rates. Although we consider a perfectly efficient storage, the problem can be generalized to a more complex model. As shown in \cite{yang}, if the complementarity constraint \eqref{eq:S_constr2} holds, the storage dynamics become 
\begin{equation}
    x_{k+1}= \alpha x_{k}-\frac{1}{\eta^+}u^+_{k}+\eta^-u^-_{k} ~~ \forall k \in \mathcal{K},
\end{equation}
where $\alpha \in (0,1]$, $\eta^+ \in (0,1]$ and $\eta^- \in (0,1]$ are the leakage coefficient, discharging and charging efficiency, respectively. In Section~\ref{sec:storageparticipation}, we derive conditions for which constraint  \eqref{eq:S_constr2} is satisfied. Further, the power constraint, the efficiency and the leakage parameters are incorporated in the storage model in the case study in Section~\ref{sec:casestudy2}. Without loss of generality, we let the storage be initially discharged, so that $\mathbf{x_0}=\mathbf{0}$ in~\eqref{eq:S_constr1c}.

\textit{Renewable Production Model:}
The renewable production is modeled as a discrete-time random process defined by $R=\{R_0,...,R_{N-1}\}$. For each time $k$, the random variable $R_k$ has a continuous and twice differentiable probability density function $f_{k}(r_k)$ and cumulative density function $F_{k}(r_k)$. For simplicity, the random variables $\{R_{k}\}$ are assumed to be mutually independent.

There are many works that consider estimating the probabilistic model for renewable production empirically from historical data. Intuitively, the uncertainty added by this data-driven approach may decrease the renewable producer's expected profit, as discussed in \cite{bitar}. Improving sensors and forecast methods to make the estimate more accurate can help mitigate this effect.

\textit{Electricity Market Model:}
We model a two-settlement market which consists of a day-ahead (DA) market followed by an imbalance settlement that penalizes uninstructed deviations.
The market is operated by an independent system operator (ISO) who is responsible for meeting the load reliably. In the DA market, all generators bid the amount of energy they are willing to commit for delivery in the next operating day. Each player also informs the ISO of their asking price, which is the minimum per-unit price they are willing to accept to deliver the amount committed. The ISO clears the market using a least-cost strategy. 

The ex-post imbalance resolution mechanism penalizes the generators that do not supply the amount of energy that they were cleared for. We focus on the bidding strategies of a renewable producer and an energy storage. Since renewable production is stochastic, this producer may be unable to meet his commitment in real-time. We assume that all renewable production exceeding the commitment is curtailed. If the renewable production is below its commitment, it pays a penalty per unit of shortfall $\lambda_p$ that is fixed and known. We note that the imbalance penalty considered can be thought of as a payment related to uninstructed deviations from dispatch instructions that is settled in real-time (e.g. as described in the CAISO market settlements \cite{caiso}). Further, the penalty values can be modeled as random variables that are assumed to be statistically independent of the renewable production of the player considered in the insurance contract by replacing the price $\lambda_p$ by its expected value.

We ignore transmission line congestion and the emergence of locational marginal prices, and thus all generators are paid a single market price. Incorporating such constraints would require a more complex mathematical model, which is left as a direction for future work. Nonetheless, we perform a case study in a constrained network in Section~\ref{sec:casestudy2} to illustrate numerically that the insurance contract proposed can be extended to a more general framework. Following other works which analyze a small subset of agents in the market, such as \cite{bitar,cho,mcconnell}, the market participants are assumed to be price-takers. It follows that the DA energy price $\mathbf{\lambda} \triangleq [\lambda_0,...,\lambda_{N-1}]^T$ can be treated as fixed and known. Finally, we assume the load to be known. All these considerations simplify the analysis and allow us to focus on the insurance contract being proposed.

\section{Problem Formulation}\label{sec:problem}

\paragraph{Utility Functions}
The utility function of each participant is his expected profit.

\textit{Baseline case:}
Both the renewable and the storage only bid in the day-ahead market and insurance contracts are not allowed. 
The expected profit of the renewable producer is
\begin{equation*}
J_{r}^{b}(\mathbf{C_{r}}) = \sum_{k=0}^{N-1} \lambda_{k}C_{rk}- E_{Rk}\left[I(C_{rk}-R_k) \lambda_{p}(C_{rk}-R_k) \right],
\end{equation*}
where $\mathbf{C_{r}}=[C_{r0},...,C_{r(N-1)}]^T \in \mathbb{R}^N$ contains the commitments for each time $k$, and $I(.)$ is the indicator function, which equals 1 if the argument is positive and 0 otherwise. For each $k$, the first term  corresponds to the revenue acquired for committing to the day-ahead market. The second term is the expected penalty due to shortage, which is taken over the random renewable production. The utility for the storage is
\begin{equation}\label{eq:Js_b}
J_{s}^{b}(\mathbf{u^+},\mathbf{u^-}) = \sum_{k=0}^{N-1} \lambda_{k}(u^+_{k}-u^-_{k})- g(u^+_{k},u^-_{k}),
\end{equation}
where, for each $k$, the first term is the revenue for supplying to and the cost for demanding from the market. The second term is the operational cost of the storage. We assume the function $g(u^+_{k},u^-_{k})$ to be convex and strictly increasing in $(u^+_{k},u^-_{k})$.

\textit{Insurance contract case:} If the storage and the renewable players are allowed to establish an insurance contract for reserve, 
the expected profit of the renewable producer 
becomes
\begin{equation}\label{eq:Jr_c}
\begin{split}
 J_{r}^{c}&(\mathbf{C_{r}},\mathbf{G_r},\mathbf{\pi_r}) = \sum_{k=0}^{N-1} \lambda_{k}C_{rk}-\pi_{rk}G_{rk}\\
&-E_{Rk}\left[I(C_{rk}-R_k-G_{rk}) \lambda_{p}(C_{rk}-R_k-G_{rk}) \right],   
\end{split}
\end{equation}
where the vectors $\mathbf{G_{r}}=[G_{r0},...,G_{r(N-1)}]^T \in \mathbb{R}^N$ and $\mathbf{\pi_{r}}=[\pi_{r0},...,\pi_{r(N-1)}]^T \in \mathbb{R}^N$ contain the reserve amounts and (per unit) prices for each time $k$. Compared to the baseline case, the renewable producer has the additional cost of the contract, and the reserve amount helps decrease the expected penalty. For the storage unit, the expected profit is
\begin{equation}\label{eq:Js_c}
\begin{split}
  J_{s}^{c}(\mathbf{u^+},&\mathbf{u^-},\mathbf{\pi_{s}}) = \sum_{k=0}^{N-1} \pi_{sk}u^+_{k}-\lambda_{k}u^-_{k}\\
&-E_{Rk}\left[g(\min(C_{rk}-R_k,u^+_{k}),u^-_{k})\right],  
\end{split}
\end{equation}
where the vector $\mathbf{\pi_{s}}=[\pi_{s0},...,\pi_{s(N-1)}]^T \in \mathbb{R}^N$ is the per unit price of reserve. The last term is the expected operational cost, which shows that the amount supplied by the storage is the lesser of the renewable shortage and the reserve in the contract. We initially consider that the storage supplies energy exclusively to the renewable and charges from the grid. However, we show in Section~\ref{sec:storageparticipation} that the storage decides on a contract through which it is only committed to supply to the renewable player at certain times of the day. Thus, for the remaining hours, the storage can trade in the electricity market in whichever opportunities are available, as long as it has the amount of energy decided in the contract at the time established. In this sense, the contract proposed serves as an additional stream of revenue for storage units in the market.

\paragraph{Contract Design Problem}
Through an insurance contract, the storage unit commits to maintaining some energy reserve available to be used in case of renewable shortage. The contract is signed ex-ante, while in real-time, the storage is called upon to supply this reserve in case of renewable shortage. If the shortage is less than the reserve established in the contract, the storage unit will supply only the amount needed to cover the shortfall; if the reserve is not enough to cover the shortage completely, the storage supplies the entire reserve and the renewable producer is responsible for paying a penalty for the shortage remaining. 

An insurance contract $\mathcal{C}$ is defined as the pair $\{\mathbf{\pi},\mathbf{G}\}$ that establishes the price per unit $\mathbf{\pi}=[\pi_{0},...,\pi_{N-1}]^T \in \mathbb{R}^N$ and amount of energy $\mathbf{G}=[G_{0},...,G_{N-1}]^T \in \mathbb{R}^N$ to be set aside as a reserve by the storage at each time $k$. We say that $\mathcal{C}$ is 
\begin{itemize}
    \item \emph{individual rational} if no participant is worse off by signing the contract, i.e. if the expected profit for any participant does not decrease when she participates in the contract;
    \item \emph{feasible} if it induces a storage policy $(\mathbf{u^+},\mathbf{u^-}) \in \mathcal{U}$ and is individual rational.
\end{itemize}

The renewable producer tries to maximize his own utility when deciding how much to bid in the day-ahead market and how much reserve to procure through an insurance contract. These decisions are made sequentially, as the contract is signed ex-ante. In the day-ahead market, the renewable player solves problem $\mathcal{P}_1$ below, where $\mathbf{G_r}$ and $\mathbf{\pi_r}$ are treated as given. For the insurance contract, the problem to be solved is $\mathcal{P}_2$.
\begin{align*} 
\mathcal{P}_1: \underset{\mathbf{C_{r}\geq 0} }{\max} \,\ J_{r}^{c}(\mathbf{C_{r}},\mathbf{G_r},\mathbf{\pi_r}), ~~ \mathcal{P}_2: \underset{\mathbf{G_r},\mathbf{\pi_r}\mathbf{\geq 0} }{\max} \,\ J_{r}^{c}(\mathbf{C_{r}^*},\mathbf{G_r},\mathbf{\pi_r})
\end{align*}

For the storage, $\mathcal{P}_3$ refers to the day-ahead market decision on how much to charge, given the amount decided to supply to the renewable. For the insurance contract ex-ante decision, $\mathcal{P}_4$ is solved. Both problems are subject to $(\mathbf{u^+},\mathbf{u^-})  \in   \mathcal{U}$.
\begin{subequations}
\begin{align*}
\mathcal{P}_3: \underset{\mathbf{u^-\geq 0}}{\max} \,\ J_{s}^{c}(\mathbf{u^+},\mathbf{u^-},\mathbf{\pi_s}), ~~ \mathcal{P}_4: \underset{\mathbf{u^+,\pi_s\geq 0} }{\max} \,\ J_{s}^{c}(\mathbf{u^+},\mathbf{u^{-*}},\mathbf{\pi_s})
\end{align*}
\end{subequations}

These problems can be easily defined for the baseline case by maintaining only the day-ahead problem, letting the reserve amounts and price be zero, and letting the storage supply to the grid instead of to the renewable producer. This formulation can also be adapted into a chance-constrained problem by adding constraints that limit the probability of renewable shortage, such as $\text{Prob}(R_k<C_{rk}^*) < \alpha$. Although that is an interesting problem, here we ignore the chance constrained formulation for simplicity.

\vspace{-0.5em}
\section{Main Results}\label{sec:results}

\subsection{Renewable Participation}
The renewable generator solves the profit-maximizing problems $\mathcal{P}_1$ and $\mathcal{P}_2$ to decide how to participate in the market.

\begin{theorem}\label{th:renewable_decisions}
The optimal renewable day-ahead bid is given~by
\begin{equation}\label{eq:cro}
    C^*_{rk} = G^*_{rk}+F^{-1}_{k} \left(\frac{\lambda_{k}}{\lambda_{p}}\right),
\end{equation}
where the optimal reserve to be purchased in the insurance contract is the maximum available if the per unit price satisfies
\begin{equation}\label{eq:upperbound}
    \pi_{rk} \leq \lambda_k,
\end{equation}
and $G^*_{rk} = 0$ otherwise.
\end{theorem}

\begin{proof}
See Appendix~\ref{app:th1}.
\end{proof}

The optimal renewable strategy reduces to that for the baseline case if we let the reserve be zero. As a buyer, the renewable producer sets an upper bound on the reserve price. If the storage asks for a price above this threshold, no contracts are feasible. Note that the renewable producer's bid increases with the amount of reserve purchased through the contract.

\subsection{Storage Participation}\label{sec:storageparticipation}
Even in the deterministic baseline case, deciding how much energy to offer is a complex task for the storage, specially due to the non-linearity in the complementarity constraint \eqref{eq:S_constr2}. The following result shows that this constraint can be relaxed.
\begin{lemma}\label{lem:uu0}
Given a strictly increasing operational cost function for the storage and considering this participant is a profit-maximizer, the complementarity constraint $u^+_{k}u^-_{k} = 0$ for all $k \in \mathcal{K}$ always holds, both in the baseline case and in the presence of an insurance contract.
\end{lemma}
\begin{proof}
See Appendix~\ref{app:lem1}.
\end{proof}

Lemma~\ref{lem:uu0} allows us to remove the non-linear constraint \eqref{eq:S_constr2}. With this, the storage baseline problem becomes a convex optimization problem. 
Let $k=\max$ (resp. $k=\min$) refer to the time with the maximum (resp. minimum) day-ahead price. Further, let $\mathbf{e_{k}}$ denote the standard $k-$th basis vector whose $k-$th entry is equal to 1 and all other entries are zero. 

\begin{theorem}\label{th:storage_decisions}
The optimal storage policy in the baseline case~is
\begin{equation*}
    \mathbf{u^-} = \mathbf{e_{\min}}\overline{E}, ~~ \mathbf{u^+} = \mathbf{e_{\max}}\overline{E},
\end{equation*}
that is, an arbitrage strategy is adopted. Further, it is individual rational for this player to sign an insurance contract following this policy if the reserve price for $k=\max$ satisfies
\begin{equation}\label{eq:lowerbound}
\begin{split}
    &\pi_{s,\max} \geq \lambda_{\max}-\frac{g(\overline{E},0)}{\overline{E}}\left(1-F_{r,\max}(C_{r,\max}-\overline{E})\right)\\
    &+\frac{1}{\overline{E}}\int_{C_{r,\max}-\overline{E}}^{C_{r,\max}}g(C_{r,\max}-R_{\max},0)f_{\max}(r)dr.
\end{split}
\end{equation}
\end{theorem}
\begin{proof}
See Appendix~\ref{app:th2}.
\end{proof}

The lower bound the storage sets on the reserve price can be interpreted as the minimum price for which the per unit expected profit earned through the contract is at least equal to the per unit expected profit that can be achieved by offering that energy in the day-ahead market. This bound depends on how likely it is that the renewable shortage is less or greater than the insurance reserve, as that determines the expected storage cost. Such dependence is expressed through the terms that include the renewable production probability functions, and the lower bound \eqref{eq:lowerbound} increases as the expected renewable shortage increases. When a shortage at least equal to the insurance reserve surely happens, that is as
\begin{equation*}
    \text{Prob}(R_{\max} \leq C_{r,\max}-\overline{E}) \to 1,
\end{equation*}
then $\pi_{s,\max} \to \lambda_{\max}$ and the storage will require at least the price that he would get in the day-ahead market.

We remark that this result can be readily extended to a storage that is not perfectly efficient at the expense of more notation. In this case, the arbitrage policy adopted would cover a wider interval of time periods, depending on the power rating and the efficiency of the storage. This device would still charge around the times when the energy price is low, and discharge around the times of high prices. Thus, an insurance contract could be signed at any of the time periods in which the storage would discharge to the grid.

\subsection{Insurance Contract Feasibility}\label{sec:feasibility}
In what follows, we analyze if the renewable and the storage players ever agree on the economic terms of the contract.

\begin{theorem}\label{th:feasible}
Let the storage adopt the arbitrage policy 
\begin{equation*}
    \mathbf{u^-} = \mathbf{e_{\min}}\overline{E}, ~~ \mathbf{u^+} = \mathbf{e_{\max}}\overline{E}.
\end{equation*}
Then, the interval of per unit reserve prices for an insurance contract to be feasible under this policy is given by
\begin{equation*}
    \mathcal{I} = [\pi_{s,\max},\lambda_{\max}],
\end{equation*}
with $\pi_{s,\max}$ as in \eqref{eq:lowerbound}, and is always non-empty.
\end{theorem}
\begin{proof}
See Appendix~\ref{app:th3}.
\end{proof}
Insurance contracts are thus always feasible for our framework, which considers an unconstrained network. For a constrained network, in Section~\ref{sec:casestudy2}, we show that feasibility will depend on the location of the players in the grid, due to the differences in locational marginal prices across the network.

The worst-case scenario for the storage unit occurs when the renewable shortage is at least as large as the reserve amount established in the contract. In this case, the storage unit will incur the operational cost of supplying the full reserve. For a given insurance contract, we can use this fact to establish a lower bound on the expected profit of the storage. 

\begin{corollary}\label{cor:feasibility}
Let the storage unit follow the arbitrage policy in Theorem~\ref{th:storage_decisions} to offer the reserve $\mathbf{u^+} = \mathbf{e_{\max}}\overline{E}$ in the insurance contract. Then, the contract with a reserve price $\mathbf{\pi}_{\max}=\mathbf{\lambda}_{\max}$ is feasible, and leads to a storage expected profit that is lower bounded by the day-ahead expected profit for the same policy.
\end{corollary}

\begin{proof}
See Appendix~\ref{app:cor1}
\end{proof}

\subsection{Storage Profitability Analysis}
We analyze whether the storage owner can be profitable as an insurance provider, even though it is not competitive in the day-ahead market. We highlight that we do not make any inferences about how profitable this energy storage may be in any services other than the day-ahead energy market. Instead, we compare the participation in the day-ahead market versus the provision of reserve to a renewable generator through an insurance contract. The following result follows from direct inspection of the expected profits for the storage unit in the baseline and in the insurance contract cases, and establishes a condition for the storage profitability in the day-ahead market.
 
\begin{theorem}\label{th:profitability}
Let $\pi$ denote the per unit price of reserve in the insurance contract and define the bounds
\begin{align}
    \underline{\Lambda} =~&1-\frac{g(0,\overline{E})}{\lambda_{\max}\overline{E}}- \frac{g(\overline{E},0)}{\lambda_{\max}\overline{E}}\\
    \overline{\Lambda} =~& \frac{\pi}{\lambda_{\max}}-\frac{g(0,\overline{E})}{\lambda_{\max}\overline{E}}-\frac{g(\overline{E},0)F_{r,\max}(C_{r,\max}-\overline{E})}{\lambda_{\max}\overline{E}}\\
    \nonumber &-\frac{1}{\lambda_{\max}\overline{E}}\int_{C_{r,\max}-\overline{E}}^{C_{r,\max}}g(C_{r,\max}-R_{\max},0)f_{\max}(r)dr
\end{align}
When adopting an arbitrage policy, the storage is profitable as an insurance provider, but not in the day-ahead market, if
\begin{equation}\label{eq:ratiocondition}
    \underline{\Lambda} \leq \frac{\lambda_{\min}}{\lambda_{\max}} < \overline{\Lambda}.
\end{equation}
\end{theorem}
\begin{proof}
Proof follows from direct inspection of the expected profits for the storage unit in the baseline and in the insurance contract cases. We simultaneously check conditions for which this player is not competitive in the day-ahead market ($J_s^{b} \leq 0$), but is profitable as an insurance provider ($J_s^{c} > 0$).
\end{proof}

Condition \eqref{eq:ratiocondition} establishes bounds on the ratio between the minimum and the maximum day-ahead prices. This price ratio is determinant to how much profit the storage can earn using arbitrage. If the lower bound holds, the difference between the maximum and minimum day-ahead prices is not high enough to cover the storage operational costs and yield a positive profit for this unit in the day-ahead market. On the other hand, the upper bound is satisfied when the payments from the insurance contract are high enough to cover both the storage operational cost and the payments made to charge from the grid.

This analysis shows that storage technologies that are still too expensive to bid in the day-ahead market may, instead, offer their energy to renewables through insurance contracts. These contracts serve as alternative sources of revenue for such storage units, keeping them from being idle when they lack competitiveness, and providing additional economic incentives for the improvement of their technology. We remark that in California, for example, it has been identified that there is a limit to the amount of storage that can provide peaking capacity according to their 4-hour rule, which credits storage units that can sustain 4 hours at maximum output \cite{denholm}. Further, the price ratio \eqref{eq:ratiocondition} increases as more peak-shaving services are provided, making it more likely that the lower bound holds, and thus more difficult for new storage operators to enter the market. Thus, we can envision that insurance contracts will allow for a higher penetration of storage in the grid by serving as an entry point in the market for peaking capacity. In Section~\ref{sec:casestudy1}, we perform a case study in which we show that the condition in Theorem~\ref{th:profitability} is not too stringent.

\subsection{Model Extension - Two-Way Contract Structure}\label{sec:modelext}
We extend our model to a full insurance contract which permits production exchange from the renewable to the storage if there is excess renewable generation. Trading this excess production is beneficial for both players. For the storage, purchasing energy at a price lower than the market price leads to a decreased cost of charging. For the renewable, selling excess energy at any price is preferable to being~curtailed. The possibility of an extra revenue may affect the renewable participation in the market, as any increase in the day-ahead commitment is now also linked to a decrease in the expected revenue from selling excess generation. To evaluate this trade-off, we modify the renewable utility \eqref{eq:Jr_c} to add the opportunity to sell excess production at the price $\pi_{ek}$ for each $k$.
\begin{equation}\label{eq:Jr_c2}
\begin{split}
 J_{r}&(\mathbf{C_{r}},\mathbf{G_r},\mathbf{\pi_r},\mathbf{\pi_e}) = \sum_{k=0}^{N-1} \lambda_{k}C_{rk}-\pi_{rk}G_{rk}\\
 &+E_{Rk}\left[I(R_k-C_{rk}) \pi_{ek}(R_k-C_{rk}) \right]\\
 &-E_{Rk}\left[I(C_{rk}-R_k-G_{rk}) \lambda_{p}(C_{rk}-R_k-G_{rk}) \right].   
\end{split}
\raisetag{0.8\baselineskip}
\end{equation}

The following results hold for every time instant $k \in \mathcal{K}$, but the time subscript is suppressed for notational simplicity.

\begin{proposition}\label{prop:extension}
Let the renewable generator sign an insurance contract with the two-way structure described in Section~\ref{sec:modelext}. Then, for  small enough $\pi_{e}$, the optimal commitment $C_r^*$ satisfies the equilibrium condition
\begin{equation}\label{eq:pr4equilibrum}
    \lambda = \lambda_pF_R\left(C_r^*-G_r\right) + \pi_e\left(1-F_R(C_r^*)\right).
\end{equation}
On the other hand, for large enough $\pi_{e}$ the player chooses to offer its maximum capacity in the day-ahead market if
\begin{equation*}
     \pi_e\mu_R \leq \lambda C_r-\pi_{r}G_{r}-E_{R}\left[I(C_{r}-R-G_{r}) \lambda_{p}(C_{r}-R-G_{r}) \right],
\end{equation*}
where $\mu_R$ is the expected renewable production. Otherwise, the renewable generator bids zero in the day-ahead market and sells all production to the storage unit.
\end{proposition}
\begin{proof}
See Appendix~\ref{app:prop1}.
\end{proof}

The above analysis assumes that the storage voluntarily agrees to purchase the excess renewable energy. Thus, $\pi_e$ can be at most the energy price $\lambda$. If the time considered is such that the renewable player has no insurance contract, but is still able to sell excess energy to the storage unit, we can set $G_r=0$ and rewrite \eqref{eq:pr4equilibrum} to find the optimal renewable day-ahead commitment as $C_r^* = F_R^{-1}\left(\frac{\lambda - \pi_e}{\lambda_p - \pi_e}\right)$. Thus, we note the opportunity cost between selling energy to the market or to the storage. The higher the reselling price $\pi_e$, the lower the day-ahead commitment. This leads to an increased likelihood of having excess energy to resell. Conversely, $C_r^*$ increases with the day-ahead energy price $\lambda$. 

\begin{figure*}[t]
\begin{multicols}{3}
    \includegraphics[width=0.9\linewidth]{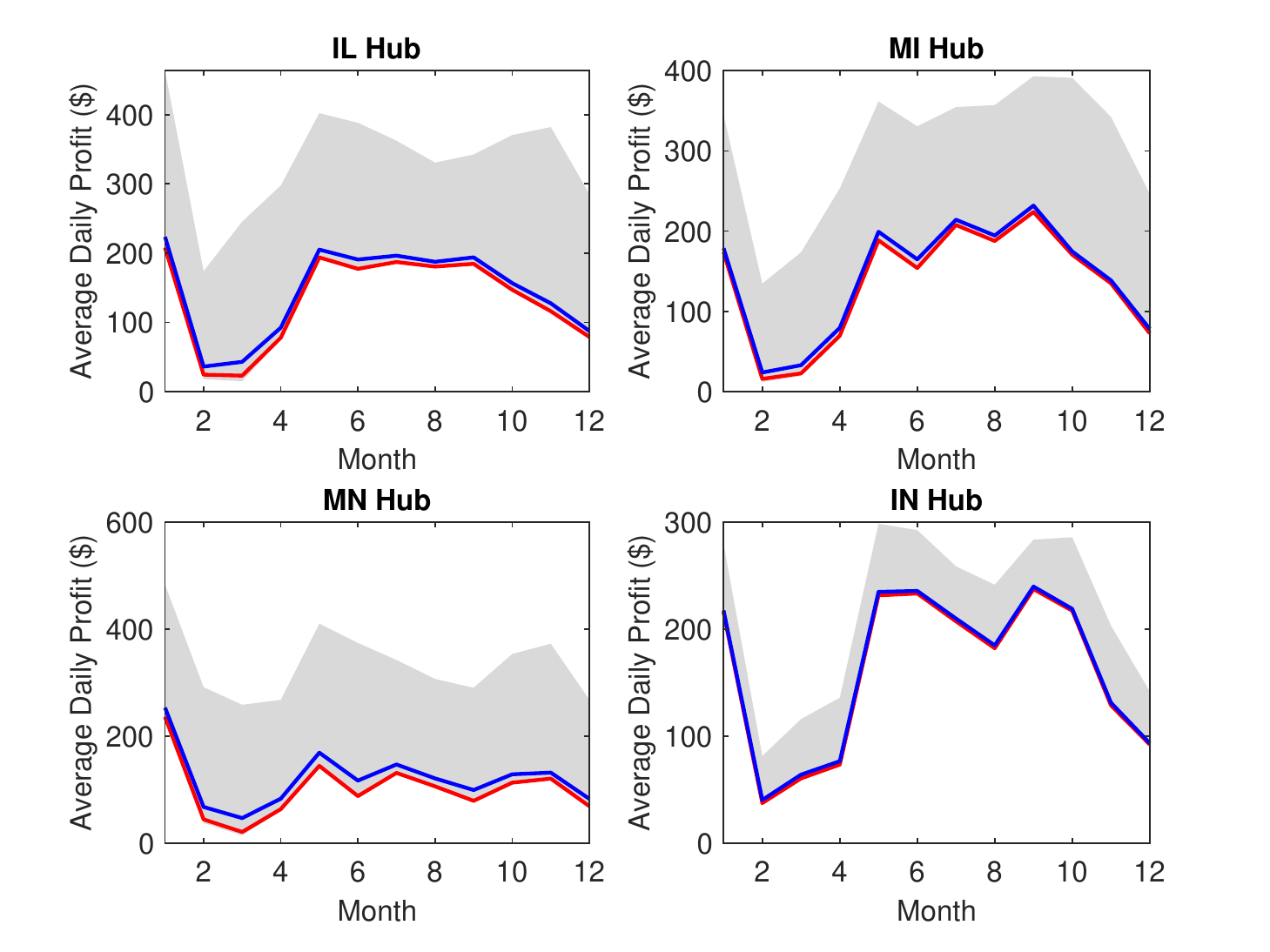}\par\caption{Storage average daily profit for baseline (red) and insurance contract (blue) cases. Shaded area shows profit variation across 1000 scenarios.}\label{fig:Js}
    \includegraphics[width=0.9\linewidth]{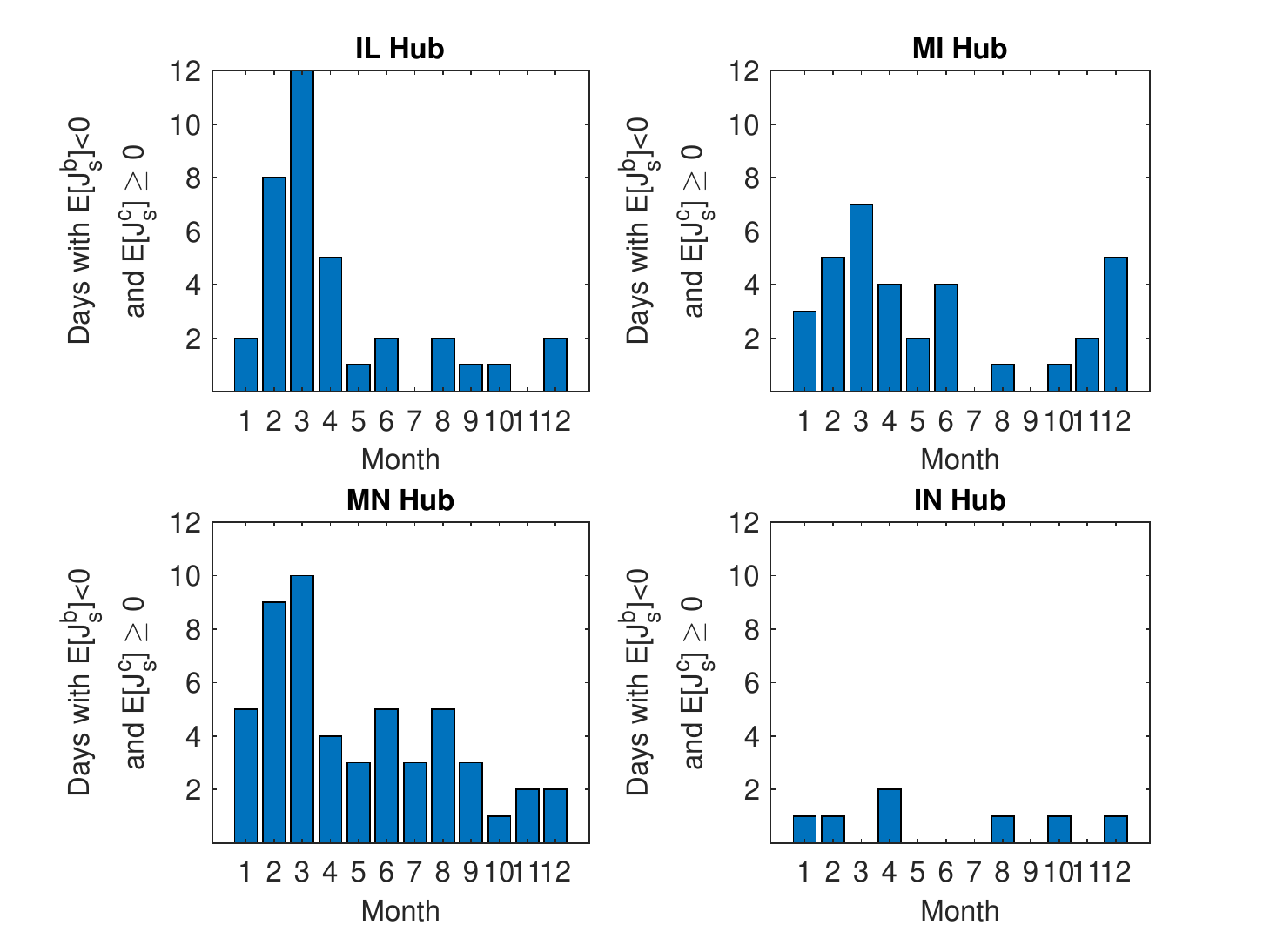}\par\caption{Number of days per month in which the storage unit is profitable as an insurance provider, but not in the day-ahead market.}\label{fig:th4}
    \includegraphics[width=0.9\linewidth]{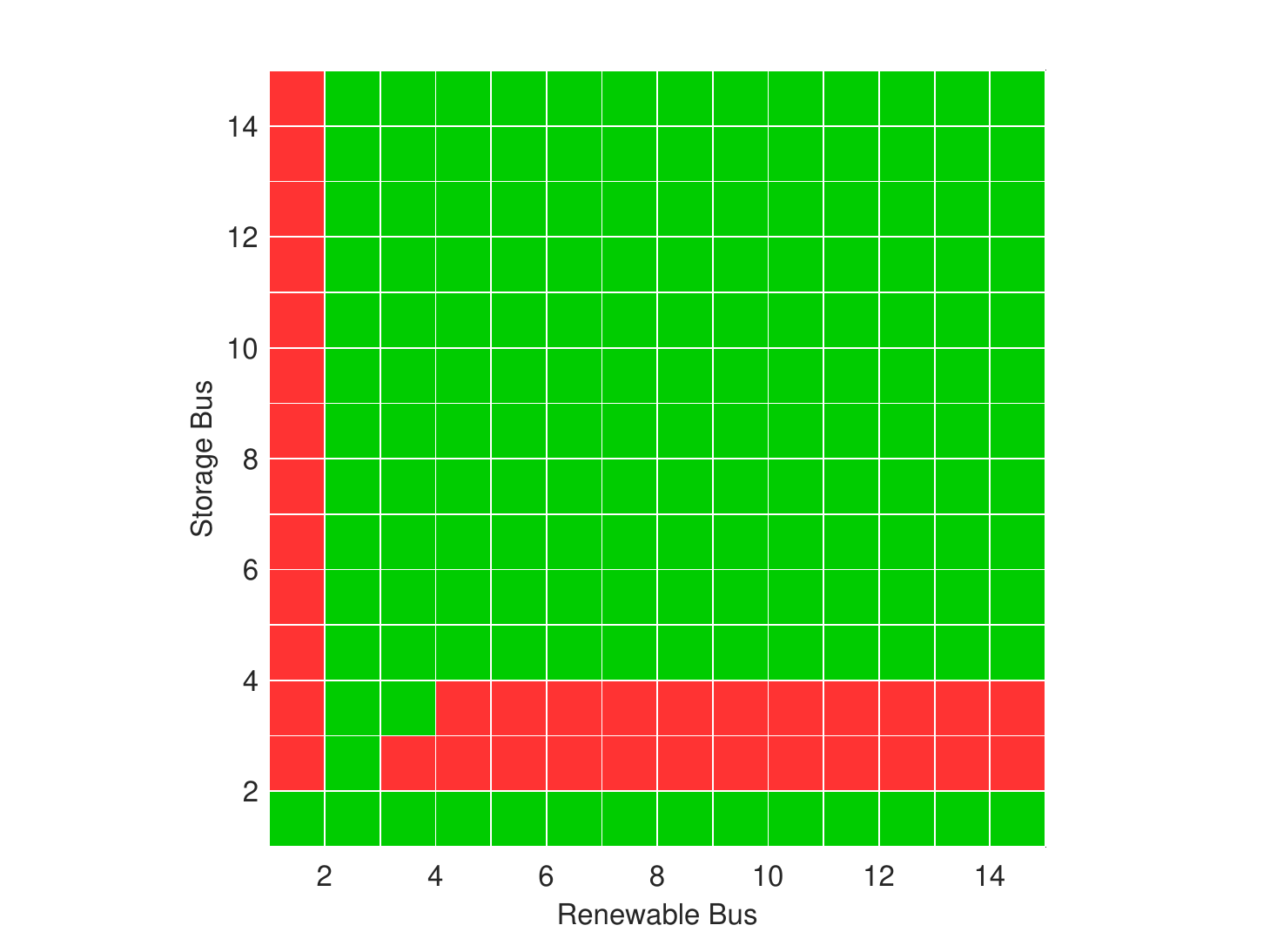}\par\caption{Insurance contract feasibility for the modified IEEE 14-bus test system.}\label{fig:ieee14bus_feasible}
\end{multicols}
\end{figure*}

\section{Case Studies}\label{sec:casestudy}
For both cases studied, we estimated the renewable production distribution from the Wind Integration National Dataset \cite{winddata,draxl}. For each month, a Gaussian distribution was fitted for the hourly wind production in each hour of the day, assuming the productions to be independent across time.

\subsection{Single-Node Case}\label{sec:casestudy1}
We consider four hubs within the Midcontinent Independent System Operator footprint -- Illinois, Michigan, Minnesota, and Indiana. We use the day-ahead prices of these hubs that refer to 2018 \cite{miso}. Four different wind productions are modeled based on each location. The storage unit investigated is a lithium-ion battery with energy capacity $\overline{E}=12\text{MWh}$ and a linear variable operation and maintenance cost of $\$7/\text{MWh}$ \cite{storage}.

We initially seek to confirm the existence of feasible insurance contracts in all four locations. We considered the contract in Corollary~\ref{cor:feasibility} and generated 1000 scenarios for renewable production to evaluate the profit achieved by the storage both when a contract is signed, and in the baseline case. Fig.~\ref{fig:Js} shows the results, where the shaded area corresponds to the profit variation observed across the 1000 scenarios generated in the presence of an insurance contract. As proved in Corollary~\ref{cor:feasibility}, the storage's expected profit with an insurance contract is lower bounded by the profit achieved in the baseline case. These profits are closer in cases with higher probabilities of renewable shortage, since this leads to the storage having to supply energy more frequently.

We also evaluated the condition in Theorem~\ref{th:profitability} to check if there are days when the storage investigated is profitable as an insurance provider, but not  in the day-ahead market. The results in Fig.~\ref{fig:th4} show that this situation happens more often in the Minnesota hub, while it is less frequent in Indiana. This happens because Minnesota is where the storage has the lowest day-ahead expected profit, leading to the highest number of days in which this unit can be profitable as an insurance provider, but not as a supplier in the day-ahead market.

Finally, we evaluated how much renewable energy was taken by the grid in real-time. Through an insurance contract, the storage commits to deliver some energy at the time with peak demand, if needed. Thus, we focused our analysis on that hour, as the behavior of the renewable generator during other times reduces to that observed in the baseline case. Fig.~\ref{fig:renewableincrease} shows that the insurance contract leads to an increase in the renewable energy share in the grid at the time of peak demand, for all hubs and all months. This is due to the change in the optimal bidding strategy of the renewable generator, which gives more room for extra production to be taken by the grid. 

\begin{figure}[htbp]
\centering
\includegraphics[width=0.4\textwidth]{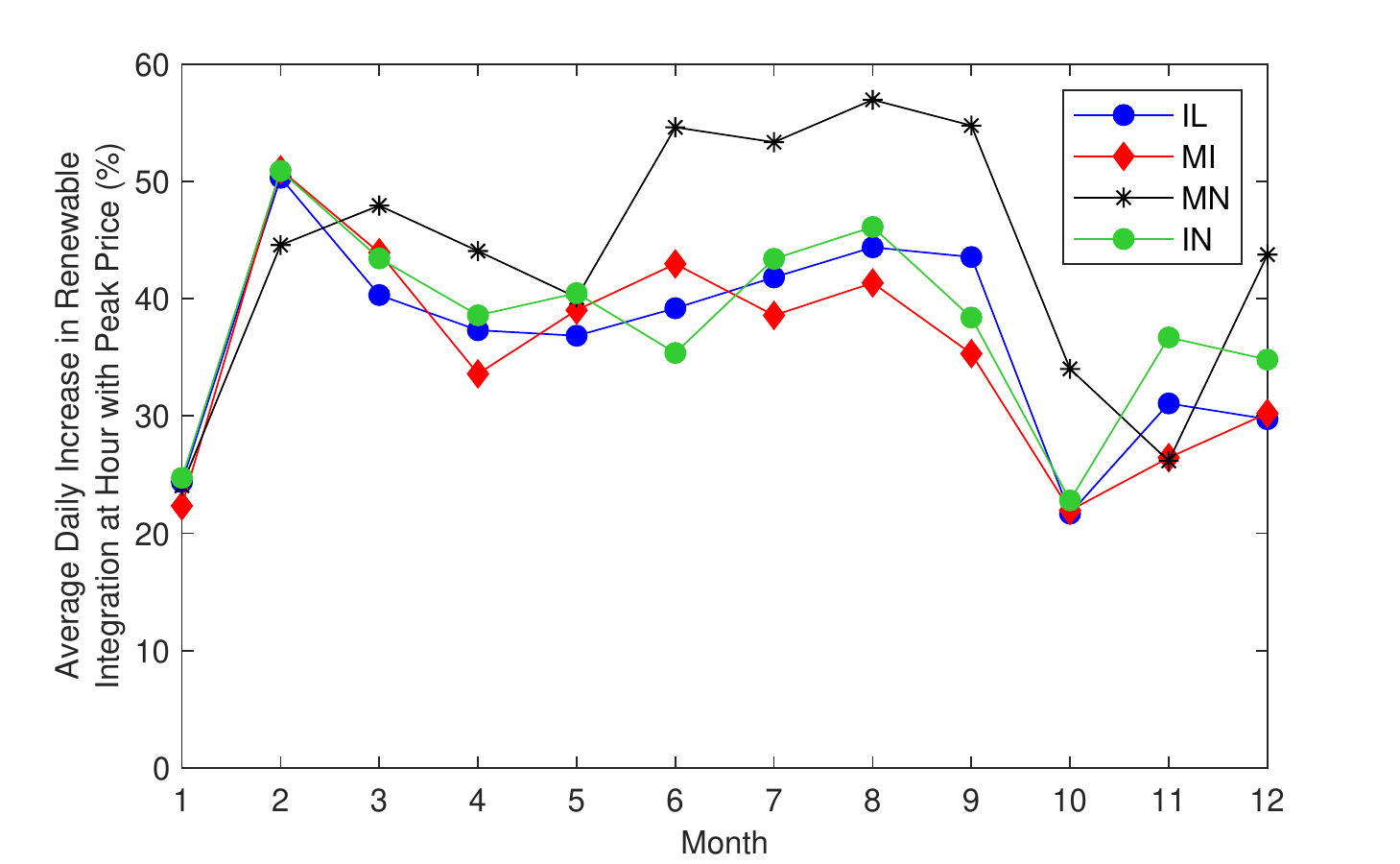}
\caption{Increase in renewable energy share at peaking demand time.}
\label{fig:renewableincrease}
\end{figure}

\subsection{Modified IEEE 14-Bus Test System Case}\label{sec:casestudy2}
We evaluate the insurance contract on the constrained IEEE 14-bus test system shown in Fig.~\ref{fig:ieee14bus}. We adopt the data from MATPOWER \cite{zimmerman} with the following modifications: (a) all transmission lines have a 80MW capacity; (b) generators at buses 1 and 2 have a 15MW ramp rate for 30min reserves; (c) a wind power plant with 32MW capacity is added; (d) a battery with energy capacity $\overline{E}=50$MWh, power capacity $\overline{P}=20$MW, linear cost $\$7/\text{MWh}$, loss factor $\alpha=0.95$, charge and discharge efficiencies $\eta^-=\eta^+=0.85$ is added.

\begin{figure}[tbp]
\centering
\includegraphics[width=0.25\textwidth]{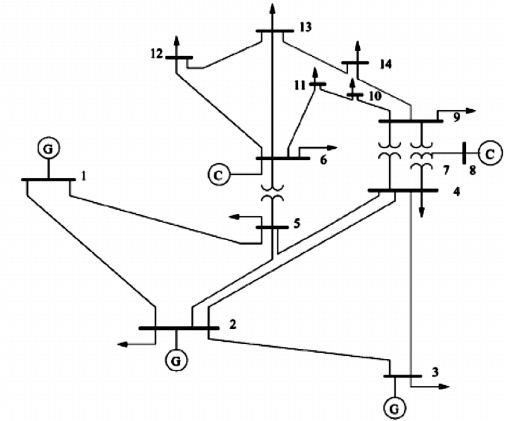}
\caption{IEEE 14-bus test system used in this case study.}
\label{fig:ieee14bus}
\end{figure}

We analyze a multi-period set-up with $N=24$ time slots corresponding to each hour of the day. The distribution for the wind production was estimated considering the month of July and the location in the Illinois hub. Further, we added a demand profile that follows a typical daily demand curve for July in this hub, which was inferred from the day-ahead price curve for this location and month. We considered the penalty for a shortage is such that the ratio $\lambda_k/\lambda_p=0.4~\forall k$. 

We solve for an optimal economic dispatch for this system. We consider 14$\times$14 scenarios, with all possible combinations of bus location for the renewable and the storage players. For each case, we let the wind producer be scheduled based on its baseline commitment and we observe the optimal schedule profile for the storage unit. Then, we determine if an insurance contract is feasible in each case, considering the storage can supply the amount of energy that it was scheduled for in the baseline case, since we know that this is a feasible trade. In this constrained network, the participants will consider their own LMPs  when deciding on the price bounds that will make it profitable for them to have an insurance contract. Fig.~\ref{fig:ieee14bus_feasible} shows the results, where each square is green if the insurance contract is feasible when the wind producer is located at bus $x$ and the storage is at bus $y$, and red otherwise.  

We observe that all squares in the diagonal, corresponding to when the renewable and the storage players are co-located, are cases with feasible insurance contracts. This result conforms with our analytical studies and the previous case study. There are cases, however, in which it is not individual rational for both players to have a contract. This is explained by the price disparity between certain nodes in this network. The storage is discharging at times with high demand, during which the grid becomes congested. The distribution of generation and load in this test case is such that node 1 has consistently the lowest LMP during congestion times, and node 2 has the highest one.

If the renewable generator is at node 1, he will set an upper bound on the contract price that is significantly lower than any other LMP; thus, if the storage is at any other node, he will not accept such a low offer, since supplying to the market at any other LMP is more profitable. This scenario represents the vertical line of red squares in Fig.~\ref{fig:ieee14bus_feasible}. Similarly, if the storage is at node 2, he will set a lower bound on the contract price that is too high for the renewable generator to accept if he is at any other node. The significant price gap between the LMP at bus 3 and at all the other buses explains the remaining cases with no feasible insurance contracts. The LMPs at all the remaining buses are close together, allowing for an insurance contract to be signed in the remaining cases. It could be argued that the storage owner should simply disregard the possibility of an insurance contract and install the storage system at the bus that is more likely to have a high LMP. However, as the number of storage units in the grid increases, the LMP peaks will decrease due to the peak shaving aspect of the storage operation in the grid. Thus, buses at which insurance contracts are currently infeasible may experience a change in this condition once the distribution of storage, as well as of other generators and loads, in the grid changes.

\section{Conclusions}\label{sec:conclusions}
We proposed an insurance contract between a renewable producer and a storage. We showed that this contract is feasible and individually rational, promotes the increase of renewable participation in the market, and serves as an extra source of revenue for storage units that are not profitable in the day-ahead market yet. 

\section*{Appendix}
\subsection{Proof of Theorem \ref{th:renewable_decisions}}\label{app:th1}
Using backwards induction, we solve for $C_{rk}^*$ taking the contract decisions as fixed. The utility function \eqref{eq:Jr_c} is concave in the commitment decisions. Using the first order conditions, 
\begin{equation*}\label{eq:dJrcdcr}
    \lambda_{k}-\lambda_p F_{Rk}\left(C_{rk}-G_{rk}\right)=0 \Rightarrow C^*_{rk} = G^*_{rk}+F^{-1}_{Rk} \left(\frac{\lambda_{k}}{\lambda_{p}}\right).
\end{equation*}

Substituting $C^*_{rk}$ back in the expected profit, we have 
\begin{align*}
J_{r}^c &= \sum_{k=0}^{N-1} \left(\lambda_{k}-\pi_{rk}\right)G_{rk}+\lambda_{k}F^{-1}_{Rk} \left(\frac{\lambda_{k}}{\lambda_{p}}\right)\\
&- E_{Rk}\left[I\left(F^{-1}_{Rk} \left(\frac{\lambda_{k}}{\lambda_{p}}\right)-R\right) \lambda_{p}\left(F^{-1}_{Rk} \left(\frac{\lambda_{k}}{\lambda_{p}}\right)-R\right) \right].
\end{align*}
This can be rewritten as $J_{r}^c = \sum_{k=0}^{N-1} \left(\lambda_{k}-\pi_{rk}\right)G_{rk}+J_{r}^{b}(\mathbf{C^*_{r}})$, where $J_{r}^{b}(\mathbf{C^*_{r}})$ is the expected renewable profit for the optimal baseline commitment. Then, it is individual rational for this player to purchase any available reserve if $\pi_{rk} \leq \lambda_k$. Otherwise, this generator is better off without the contract.

\subsection{Proof of Lemma \ref{lem:uu0}}\label{app:lem1}
Let $(\mathbf{u^+},\mathbf{u^-})$ and $(\mathbf{\tilde{u}^+},\mathbf{\tilde{u}^-})$ be storage policies such that 
\begin{align}
    \label{eq:uu0}-u^+_{k}+u^-_{k} &= u_k, ~ u^+_{k}u^-_{k} = 0\\
    \label{eq:uug0}-\tilde{u}^+_{k}+\tilde{u}^-_{k} &= u_k, ~ \tilde{u}^+_{k}\tilde{u}^-_{k} > 0
\end{align}
Following \cite[Theorem 1]{yang}, we can show that $u^+_{k}<\tilde{u}^+_{k}$ and $u^-_{k}<\tilde{u}^-_{k}$. Therefore, for a strictly increasing cost function,
\begin{equation}\label{eq:costineq}
    g(\min(C_{rk}-R_k,u^+_{k}),u^-_{k}) \leq g(u^+_{k},u^-_{k}) < g(\tilde{u}^+_{k},\tilde{u}^-_{k}).
\end{equation}
Let $J_s$ and $\tilde{J}_s$ be the storage profit under these policies. The storage is better off with the policy satisfying the complementarity constraint if and only if $J_s>\tilde{J}_s$. 
From \eqref{eq:uu0} and \eqref{eq:uug0}, $\tilde{u}^+_{k}-u^+_{k}=\tilde{u}^-_{k}-u^-_{k} := U_k > 0.$

In the presence of an insurance contract, and for the case $\min(C_{rk}-R_k,u^+_{k})=u^+_{k}$, the inequality $J_s>\tilde{J}_s$ gives us
\begin{align}
    \nonumber  g(\tilde{u}^+_{k},\tilde{u}^-_{k})-g(u^+_{k},u^-_{k}) &> \pi_{sk}(\tilde{u}^+_{k}-u^+_{k})-\lambda_{k}(\tilde{u}^-_{k}-u^-_{k}) \\
    \label{eq:condcompl}g(\tilde{u}^+_{k},\tilde{u}^-_{k})-g(u^+_{k},u^-_{k}) &> (\pi_{sk}-\lambda_{k})U_{k}.
\end{align}
From \eqref{eq:costineq}, we note that the left hand side of \eqref{eq:condcompl} is positive. Thus, this condition always holds if $\pi_{sk} \leq \lambda_{k}$. This inequality coincides with the upper bound set by the renewable on the reserve price \eqref{eq:upperbound}, so it must hold when a contract is signed. Thus, in the presence of an insurance contract, the policy satisfying $u^+_{k}u^-_{k} = 0$ will be chosen and the constraint \eqref{eq:S_constr2} can be relaxed. It is straightforward to note that this also holds if $\min(C_{rk}-R_k,u^+_{k})=C_{rk}-R_k$, as well as for the baseline case, for which the right hand side of \eqref{eq:condcompl} will be zero.

\subsection{Proof of Theorem \ref{th:storage_decisions}}\label{app:th2}
Since the storage baseline problem is convex, an optimal solution will satisfy the Karush-Kuhn-Tucker (KKT) conditions. Let $\mathbf{\underline{\mu}}=[\underline{\mu}_{0},...,\underline{\mu}_{(N-1)}]^T$ and $\mathbf{\overline{\mu}}=[\overline{\mu}_{0},...,\overline{\mu}_{(N-1)}]^T$ be the Lagrange multipliers corresponding to the lower and upper storage energy constraints \eqref{eq:S_constr1c}, respectively, and the multipliers $\overline{\rho_k}$ and $\underline{\rho_k}$ refer to the positivity constraint \eqref{eq:S_constr3}. Further, let $\mathbf{A^{+T}_k}$ and $\mathbf{A^{-T}_k}$ be the transpose of the $k$-th column of $\mathbf{A^{+}}$ and $\mathbf{A^{-}}$. For every $k$, the KKT conditions are given by
\begin{align}
    \label{eq:kkt_s1}\lambda_k-\frac{\partial}{\partial u^+_k}g(u^+_{k},u^-_{k})-\mathbf{A^{+T}_k\underline{\mu}}+\mathbf{A^{-T}_k\overline{\mu}}+\overline{\rho_k} &= 0,\\
    \label{eq:kkt_s2}-\lambda_k-\frac{\partial}{\partial u^-_k}g(u^+_{k},u^-_{k})+\mathbf{A^{+T}_k\underline{\mu}}-\mathbf{A^{-T}_k\overline{\mu}}+\underline{\rho_k} &= 0, \\
    \label{eq:kkt_c1}\mathbf{\underline{\mu}} \circ  (\mathbf{A^+u^+ + A^-u^-}) &= \mathbf{0},\\
     \label{eq:kkt_c2}\mathbf{\overline{\mu}} \circ (\mathbf{\overline{E}-A^+u^+-A^-u^-}) &= \mathbf{0},
\end{align}
where, $A \circ B$ denotes the element-wise product of matrices $A$ and $B$. Lastly, all multipliers must be non-negative. Proof that an arbitrage policy is optimal in the baseline case follows from direct inspection of \eqref{eq:kkt_s1}--\eqref{eq:kkt_c2} for $\mathbf{u^+} = \mathbf{e_{\max}}\overline{E}$ and $\mathbf{u^-}=\mathbf{e_{\min}}\overline{E}$. If the storage signs the insurance contract and offers $\overline{E}$ at time $k=\max$, his expected profit is
\begin{equation}\label{eq:jsc_th2}
\begin{split}
    J_s^{c} = (&\pi_{s,\max}-\lambda_{\min})\overline{E}-g(0,\overline{E})\\
    &-E_{R,\max}[g(\min(C_{r,\max}-R_{\max},\overline{E}),0)].
\end{split}
\end{equation}
The last term in \eqref{eq:jsc_th2} can be rewritten as $\int_{0}^{C_{r,\max}-\overline{E}}g(\overline{E},0)f_{\max}(r)dr
+\int_{C_{r,\max}-\overline{E}}^{C_{r,\max}}g(C_{r,\max}-R_{\max},0)f_{\max}(r)dr$. If he opts out and only bids in the day-ahead market instead,
\begin{equation}\label{eq:jsda_th2}
    J_s^{b} = (\lambda_{\max}-\lambda_{\min})\overline{E}-g(0,\overline{E})-g(\overline{E},0).
\end{equation}
Individual rationality holds if $J_s^{c}\geq J_s^{b}$. Using the expressions above and recognizing that $\int_{0}^{C_{r,\max}-\overline{E}}g(\overline{E},0)f_{\max}(r)dr = g(\overline{E},0)F_{r,\max}(C_{r,\max}-\overline{E})$
yields the condition in Theorem~\ref{th:storage_decisions}.

\subsection{Proof of Theorem \ref{th:feasible}}\label{app:th3}
The upper bound \eqref{eq:upperbound} on the price is greater than or equal to the lower bound \eqref{eq:lowerbound} for $k=\max$ if
\begin{equation}\label{eq:nonempty}
\begin{split}
    &g(\overline{E},0)\left(1-F_{r,\max}(C_{r,\max}-\overline{E})\right) \geq \\
    &\int_{C_{r,\max}-\overline{E}}^{C_{r,\max}}g(C_{r,\max}-R_{\max},0)f_{\max}(r)dr.
\end{split}
\end{equation}
The right-hand side integral can be bounded from above by the quantity $g(\overline{E},0)\left(F_{r,\max}(C_{r,\max})-F_{r,\max}(C_{r,\max}-\overline{E})\right),$ and further, $g(\overline{E},0)\left(1-F_{r,\max}(C_{r,\max}-\overline{E})\right)$, which is the same as the left-hand side of \eqref{eq:nonempty}. Thus, the condition for the reserve price interval to be non-empty holds.

\subsection{Proof of Corollary \ref{cor:feasibility}}\label{app:cor1}
The arbitrage policy is feasible, and, from Theorem~\ref{th:feasible}, the contract price proposed is the upper bound of the interval that guarantees individual rationality for both players. Then, this insurance contract is feasible. To find the lower bound on the profit, we analyze the worst-case scenario. In this case, $\min(C_{r,\max}-R_{\max},\overline{E})=\overline{E}$, which leads to $J_s^{b}=J_s^{c}$.

\subsection{Proof of Proposition \ref{prop:extension}}\label{app:prop1}
The first and second order derivatives of $J_{r}$ are
\begin{align}
    \label{eq:dJr}\frac{\partial J_r}{\partial C_r} &= \lambda - \lambda_pF_R\left(C_r-G_r\right) - \pi_e\left(1-F_R(C_r)\right)\\
    \label{eq:d2Jr}\frac{\partial^2 J_r}{\partial C_r^2} &= - \lambda_pf_R\left(C_r-G_r\right) + \pi_ef_R(C_r)
\end{align}
Note that \eqref{eq:d2Jr} will be negative (resp. positive) for small (resp. large) enough $\pi_{e}$, which means the function is concave (resp. convex). Further, condition \eqref{eq:pr4equilibrum} is found by setting \eqref{eq:dJr} to zero. In the convex case, the optimal bid is found by checking when the expected utility is maximized at each boundary.


\ifCLASSOPTIONcaptionsoff
  \newpage
\fi



\bibliographystyle{IEEEtran}
\bibliography{references}

\begin{thebibliography}{10}
\providecommand{\url}[1]{#1}
\csname url@samestyle\endcsname
\providecommand{\newblock}{\relax}
\providecommand{\bibinfo}[2]{#2}
\providecommand{\BIBentrySTDinterwordspacing}{\spaceskip=0pt\relax}
\providecommand{\BIBentryALTinterwordstretchfactor}{4}
\providecommand{\BIBentryALTinterwordspacing}{\spaceskip=\fontdimen2\font plus
\BIBentryALTinterwordstretchfactor\fontdimen3\font minus
  \fontdimen4\font\relax}
\providecommand{\BIBforeignlanguage}[2]{{%
\expandafter\ifx\csname l@#1\endcsname\relax
\typeout{** WARNING: IEEEtran.bst: No hyphenation pattern has been}%
\typeout{** loaded for the language `#1'. Using the pattern for}%
\typeout{** the default language instead.}%
\else
\language=\csname l@#1\endcsname
\fi
#2}}
\providecommand{\BIBdecl}{\relax}
\BIBdecl

\bibitem{divya}
K.~Divya and J.~{{\O}}stergaard, ``Battery energy storage technology for power
  systems—an overview,'' \emph{Electric Power Systems Research}, vol.~79,
  no.~4, pp. 511 -- 520, 2009.

\bibitem{xu2016}
{Bolun Xu}, Y.~{Dvorkin}, D.~S. {Kirschen}, C.~A. {Silva-Monroy}, and
  J.~{Watson}, ``A comparison of policies on the participation of storage in
  u.s. frequency regulation markets,'' in \emph{2016 IEEE Power and Energy
  Society General Meeting (PESGM)}, July 2016, pp. 1--5.

\bibitem{pjm_stateofmarket}
``Quarterly state of the market report for pjm: January through march,''
  Monitoring Analytics, LLC, Tech. Rep., May 2019.

\bibitem{denholm}
P.~Denholm and R.~Margolis, ``The potential for energy storage to provide
  peaking capacity in california under increased penetration of solar
  photovoltaics,'' National Renewable Energy Laboratory, Tech. Rep.
  NREL/TP-6A20-70905, March 2018.

\bibitem{kalathil2017sharing}
D.~Kalathil, C.~Wu, K.~Poolla, and P.~Varaiya, ``The sharing economy for the
  electricity storage,'' \emph{IEEE Transactions on Smart Grid}, 2017.

\bibitem{bpa}
\emph{Bonneville Power Administration, Self Supply of Balancing Services}, Bpa
  transmission business practice, version 3~ed., 2017.

\bibitem{porter}
K.~Porter, K.~Starr, and A.~Mills, ``Variable generation and electricity
  markets. a living summary of markets and market rules for variable generation
  in north america,'' Utility Variable Generation Integration Group (UVIG),
  Tech. Rep., 2015.

\bibitem{bitar}
E.~Y. Bitar, R.~Rajagopal, P.~P. Khargonekar, K.~Poolla, and P.~Varaiya,
  ``Bringing wind energy to market,'' \emph{IEEE Transactions on Power
  Systems}, vol.~27, no.~3, pp. 1225--1235, Aug 2012.

\bibitem{krishnamurthy}
D.~Krishnamurthy, C.~Uckun, Z.~Zhou, P.~R. Thimmapuram, and A.~Botterud,
  ``Energy storage arbitrage under day-ahead and real-time price uncertainty,''
  \emph{IEEE Transactions on Power Systems}, vol.~33, no.~1, pp. 84--93, Jan
  2018.

\bibitem{wangy}
Y.~Wang, Y.~Dvorkin, R.~Fernández-Blanco, B.~Xu, T.~Qiu, and D.~S. Kirschen,
  ``Look-ahead bidding strategy for energy storage,'' \emph{IEEE Trans.
  Sustainable Energy}, vol.~8, no.~3, pp. 1106--1117, July 2017.

\bibitem{xu}
B.~Xu, J.~Zhao, T.~Zheng, E.~Litvinov, and D.~S. Kirschen, ``Factoring the
  cycle aging cost of batteries participating in electricity markets,''
  \emph{IEEE Trans. Power Systems}, vol.~33, no.~2, pp. 2248--2259, March 2018.

\bibitem{baggu}
M.~M. {Baggu}, A.~{Nagarajan}, D.~{Cutler}, D.~{Olis}, T.~O. {Bialek}, and
  M.~{Symko-Davies}, ``Coordinated optimization of multiservice dispatch for
  energy storage systems with degradation model for utility applications,''
  \emph{IEEE Transactions on Sustainable Energy}, vol.~10, no.~2, pp. 886--894,
  April 2019.

\bibitem{pandzic}
H.~Pand\v{z}i\'{c}, Y.~Dvorkin, and M.~Carri\'{o}n, ``Investments in merchant
  energy storage: Trading-off between energy and reserve markets,''
  \emph{Applied Energy}, vol. 230, pp. 277 -- 286, 2018.

\bibitem{zou}
P.~{Zou}, Q.~{Chen}, Q.~{Xia}, G.~{He}, and C.~{Kang}, ``Evaluating the
  contribution of energy storages to support large-scale renewable generation
  in joint energy and ancillary service markets,'' \emph{IEEE Transactions on
  Sustainable Energy}, vol.~7, no.~2, pp. 808--818, April 2016.

\bibitem{mcconnell}
D.~McConnell, T.~Forcey, and M.~Sandiford, ``Estimating the value of
  electricity storage in an energy-only wholesale market,'' \emph{Applied
  Energy}, vol. 159, pp. 422 -- 432, 2015.

\bibitem{baker}
K.~{Baker}, G.~{Hug}, and X.~{Li}, ``Energy storage sizing taking into account
  forecast uncertainties and receding horizon operation,'' \emph{IEEE
  Transactions on Sustainable Energy}, vol.~8, no.~1, pp. 331--340, Jan 2017.

\bibitem{ghofrani}
M.~{Ghofrani}, A.~{Arabali}, M.~{Etezadi-Amoli}, and M.~S. {Fadali}, ``A
  framework for optimal placement of energy storage units within a power system
  with high wind penetration,'' \emph{IEEE Transactions on Sustainable Energy},
  vol.~4, no.~2, pp. 434--442, April 2013.

\bibitem{thrampoulidis}
C.~{Thrampoulidis}, S.~{Bose}, and B.~{Hassibi}, ``Optimal placement of
  distributed energy storage in power networks,'' \emph{IEEE Transactions on
  Automatic Control}, vol.~61, no.~2, pp. 416--429, Feb 2016.

\bibitem{kim}
\BIBentryALTinterwordspacing
J.~H. Kim and W.~B. Powell, ``Optimal energy commitments with storage and
  intermittent supply,'' \emph{Operations Research}, vol.~59, no.~6, pp.
  1347--1360, 2011. [Online]. Available:
  \url{https://doi.org/10.1287/opre.1110.0971}
\BIBentrySTDinterwordspacing

\bibitem{su}
H.~Su and A.~E. Gamal, ``Modeling and analysis of the role of energy storage
  for renewable integration: Power balancing,'' \emph{IEEE Transactions on
  Power Systems}, vol.~28, no.~4, pp. 4109--4117, Nov 2013.

\bibitem{teleke}
S.~{Teleke}, M.~E. {Baran}, S.~{Bhattacharya}, and A.~Q. {Huang}, ``Optimal
  control of battery energy storage for wind farm dispatching,'' \emph{IEEE
  Trans. Energy Conversion}, vol.~25, no.~3, pp. 787--794, Sep. 2010.

\bibitem{kaabeche}
\BIBentryALTinterwordspacing
A.~Kaabeche, M.~Belhamel, and R.~Ibtiouen, ``Sizing optimization of
  grid-independent hybrid photovoltaic/wind power generation system,''
  \emph{Energy}, vol.~36, no.~2, pp. 1214 -- 1222, 2011. [Online]. Available:
  \url{http://www.sciencedirect.com/science/article/pii/S0360544210006699}
\BIBentrySTDinterwordspacing

\bibitem{asensio}
M.~{Asensio} and J.~{Contreras}, ``Risk-constrained optimal bidding strategy
  for pairing of wind and demand response resources,'' \emph{IEEE Transactions
  on Smart Grid}, vol.~8, no.~1, pp. 200--208, 2017.

\bibitem{li}
Y.~{Li}, Z.~{Yang}, G.~{Li}, D.~{Zhao}, and W.~{Tian}, ``Optimal scheduling of
  an isolated microgrid with battery storage considering load and renewable
  generation uncertainties,'' \emph{IEEE Transactions on Industrial
  Electronics}, vol.~66, no.~2, pp. 1565--1575, 2019.

\bibitem{yang}
P.~Yang and A.~Nehorai, ``Joint optimization of hybrid energy storage and
  generation capacity with renewable energy,'' \emph{IEEE Transactions on Smart
  Grid}, vol.~5, no.~4, pp. 1566--1574, July 2014.

\bibitem{caiso}
``California independent system operator settlements and billings,'' available
  \url{http://www.caiso.com/Documents/Section11-CAISOSettlements-Billing-asof-Jan1-2019.pdf}.

\bibitem{cho}
J.~Cho and A.~N. Kleit, ``Energy storage systems in energy and ancillary
  markets: A backwards induction approach,'' \emph{Applied Energy}, vol. 147,
  pp. 176 -- 183, 2015.

\bibitem{winddata}
C.~Draxl, B.-M. Hodge, A.~Clifton, and J.~McCaa, ``Overview and meteorological
  validation of the wind integration national dataset toolkit,'' National
  Renewable Energy Laboratory, Tech. Rep. NREL/TP-5000-61740, April 2015.

\bibitem{draxl}
C.~Draxl, A.~Clifton, B.-M. Hodge, and J.~McCaa, ``The wind integration
  national dataset toolkit,'' \emph{Applied Energy}, vol. 151, pp. 355 -- 366,
  2015.

\bibitem{miso}
\BIBentryALTinterwordspacing
``Midcontinent independent system operator market reports,'' accessed: October
  2019. [Online]. Available:
  \url{https://www.misoenergy.org/markets-and-operations/market-reports/}
\BIBentrySTDinterwordspacing

\bibitem{storage}
V.~Viswanathan, M.~Kintner-Meyer, P.~Balducci, and C.~Jin, ``National
  assessment of energy storage for grid balancing and arbitrage: Cost and
  performance characterization,'' PNNL, Tech. Rep. PNNL-21388, March 2013.

\bibitem{zimmerman}
R.~D. {Zimmerman}, C.~E. {Murillo-Sanchez}, and R.~J. {Thomas}, ``Matpower:
  Steady-state operations, planning, and analysis tools for power systems
  research and education,'' \emph{IEEE Transactions on Power Systems}, vol.~26,
  no.~1, pp. 12--19, Feb 2011.

\end{thebibliography}

%

\begin{IEEEbiography}[{\includegraphics[width=1in,height=1.25in,clip,keepaspectratio]{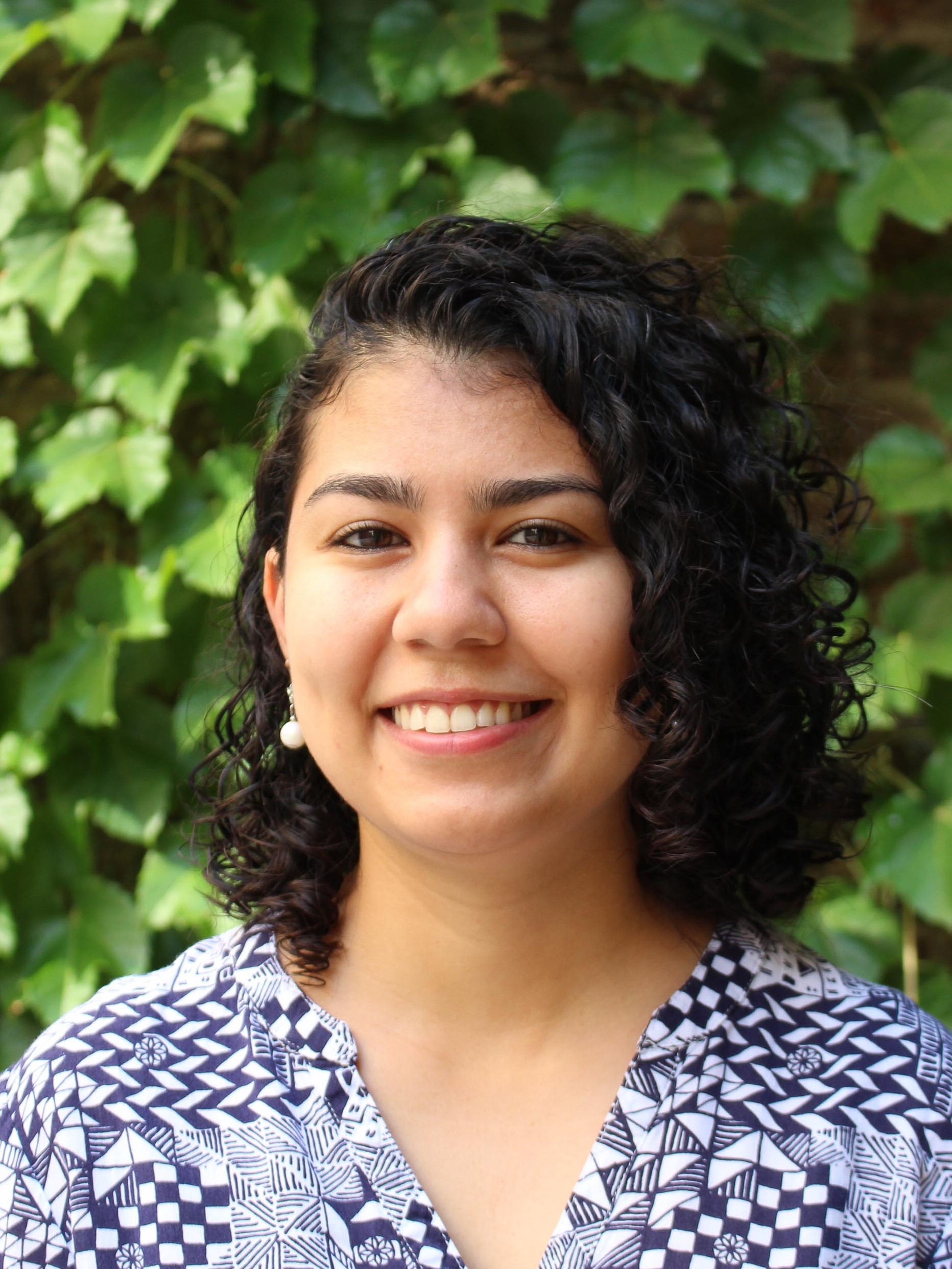}}]{Nayara Aguiar} (S’10–GS’16) received the B.Sc. degree in electrical engineering from the Federal University of Campina Grande, Brazil, in 2016, and the M.S. degree in electrical engineering from the University of Notre Dame in 2018, where she is currently pursuing the Ph.D. degree with the Department of Electrical Engineering. Her current research interests include design and analysis of electricity markets in the presence of intermittent renewable energy generation. She was a recipient of the 2019 Patrick and Jana Eilers Graduate Student Fellowship for Energy Related Research Electrical Engineering from the Center for Sustainable Energy at Notre Dame. 
\end{IEEEbiography}

\begin{IEEEbiography}[{\includegraphics[width=1in,height=1.25in,clip,keepaspectratio]{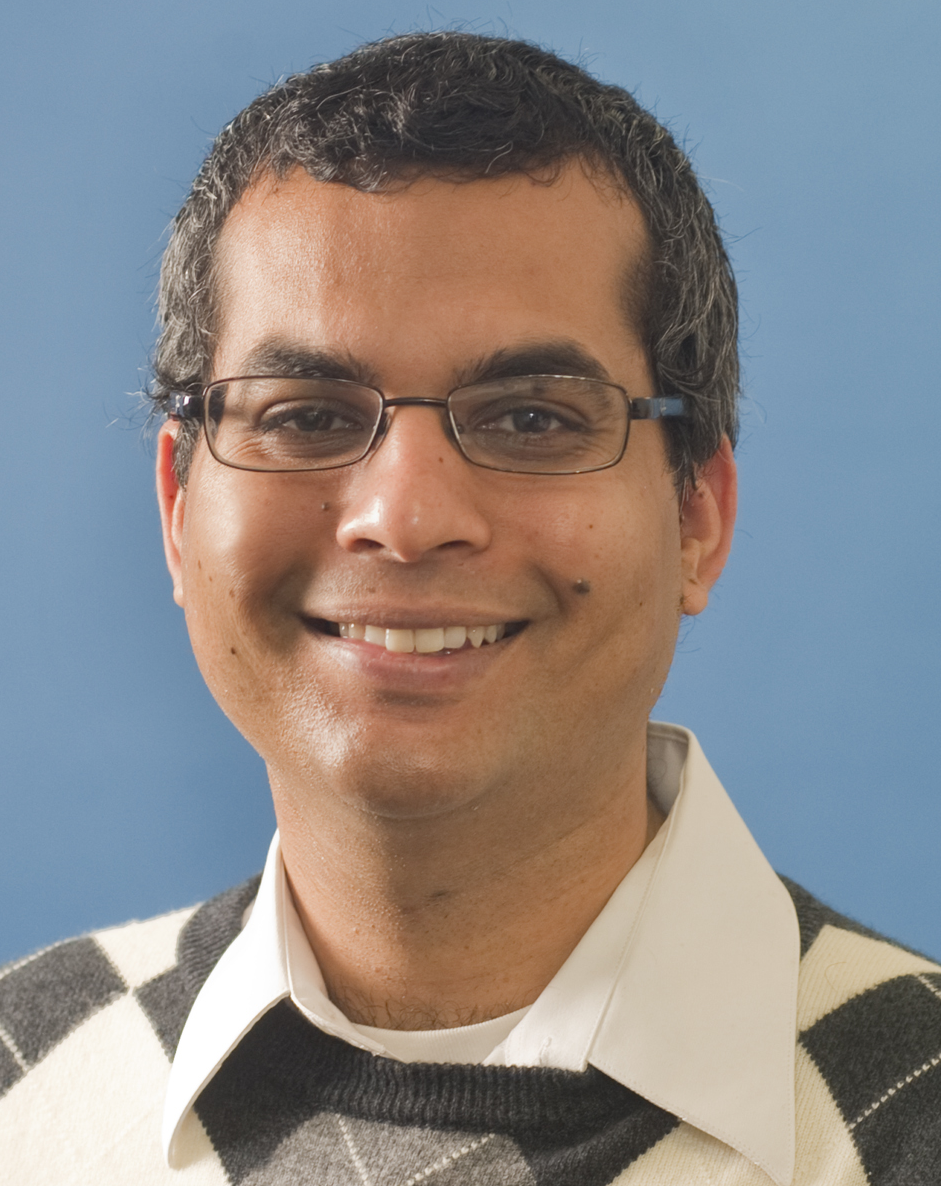}}]{Vijay Gupta} (M’07–SM’18) is in the Department of Electrical Engineering at the University of Notre Dame. He received his B. Tech degree at Indian Institute of Technology, Delhi, and his M.S. and Ph.D. at California Institute of Technology. He received the 2018 Antonio J Ruberti Award from the IEEE Control Systems Society, the 2013 Donald P. Eckman Award from the American Automatic Control Council and a 2009 National Science Foundation (NSF) CAREER Award. His research and teaching interests are broadly in the interface of communication, control, distributed computation, and human decision making. 
\end{IEEEbiography}








\end{document}